\documentclass[11pt,a4paper,twoside]{article}
\pdfoutput=1

\usepackage{amsmath}
\usepackage{amsfonts}
\usepackage{amsthm}
\usepackage{amssymb}
\usepackage[margin=2.5cm]{geometry}
\usepackage{microtype}
\usepackage{graphicx}
\usepackage{array}

\author{Ralph Christian Bottesch\thanks{CWI, Amsterdam, the Netherlands. Supported by ERC Consolidator Grant QPROGRESS 615307.}}
\title{Relativization and Interactive Proof Systems in Parameterized Complexity Theory}

\theoremstyle{definition}
\newtheorem{theor}{Theorem}
\theoremstyle{definition}

\theoremstyle{definition}

\theoremstyle{definition}
\newtheorem{corol}[theor]{Corollary}
\theoremstyle{definition}
\newtheorem{rem}[theor]{Remark}
\theoremstyle{definition}
\newtheorem{fact}[theor]{Fact}
\theoremstyle{definition}

\theoremstyle{definition}
\newtheorem{defin}[theor]{Definition}
\theoremstyle{definition}

\theoremstyle{definition}

\def\R{\textrm{R}}
\def\poly{\textrm{poly}}

\begin{document}
\maketitle

\begin{abstract}
We introduce some classical complexity-theoretic techniques to Parameterized Complexity. First, we study relativization for the machine models that were used by Chen, Flum, and Grohe (2005) to characterize a number of parameterized complexity classes. Here we obtain a new and non-trivial characterization of the $\mathbf{A}$-Hierarchy in terms of oracle machines, and parameterize a famous result of Baker, Gill, and Solovay (1975), by proving that, relative to specific oracles, $\mathbf{FPT}$ and $\mathbf{A[1]}$ can either coincide or differ (a similar statement holds for $\mathbf{FPT}$ and $\mathbf{W[P]}$). Second, we initiate the study of interactive proof systems in the parameterized setting, and show that every problem in the class $\mathbf{AW[SAT]}$ has a proof system with ``short'' interactions, in the sense that the number of rounds is upper-bounded in terms of the parameter value alone.
\end{abstract}

\section{Introduction}
In Parameterized Complexity Theory, the complexity of computational problems is measured not only in terms of the size of the input, $|x|$, but also in terms of a parameter $k$ which measures some additional structure of the input. This approach is justified by two observations:

1.\ that a running time of $2^k |x|$ is preferable to $|x|^k$, although both are exponential in $k$, because the former means that an algorithm can be run on large instances if $k$ is not too large;

2.\ that even problems which are believed to be hard can have algorithms that run in time $\exp(k)\poly(|x|)$, for very natural choices of the parameter.

These observations motivate the study of $\mathbf{FPT}$, the class of \emph{fixed-parameter tractable} problems (which can be solved in time $f(k)\poly(|x|)$, for some computable function $f$). To this relaxed notion of computational tractability there corresponds a matching notion of intractability.

The complexity classes capturing parameterized intractability were originally defined as closures, under suitably defined parameterized reductions, of specific problems that were conjectured to not have fpt-algorithms (see \cite{df1}, or the more recent \cite{df2}). This approach ensured that most of these ``hard'' classes contained an interesting or somewhat natural complete problem, and, in the case of $\mathbf{W}[1]$, produced a ``web of reductions'' similar to the one for $\mathbf{NP}$-complete problems in classical complexity.

However, defining complexity classes only via reductions to specific problems means that the resulting classes may not have characterizations in terms of computing machines, or, indeed, any natural characterizations except the definition.  This in turn can mean that many proof techniques from classical complexity are not usable in the parameterized setting, because they rely on different characterizations that do not apply to any one parameterized complexity class. To give an example, in the proof of  $\mathbf{IP}=\mathbf{PSPACE}$ (\cite{sham}, see also \cite{shen}), both the definition of $\mathbf{PSPACE}$ in terms of space-bounded computation, and the characterization of this class in terms of alternating polynomial-time computation are used. In the parameterized world, this equivalence between space and alternating time seems to break down \cite{cfg1}, and parameterized interactive proof systems do not appear to have been studied at all, so no similar theorem is known in this setting.

Surprisingly (given the way they were originally defined), many of the classes capturing parameterized intractability turned out to have characterizations in terms of computing machines: In three papers, Chen \cite{cfg1, cfg2, cfg3}, Flum \cite{cfg1, cfg2, cfg3}, and Grohe \cite{cfg1, cfg3} showed that certain kinds of nondeterministic random access machines (RAMs) exactly define some important parameterized classes:

-$\mathbf{W[P]}$ and $\mathbf{AW[P]}$ are characterized by RAMs that can nondeterministically \footnote{Throughout this paper, nondeterminism will mean alternating nondeterminism with a number of alternations that will be clear from the context. This should not cause any confusion, since simple nondeterminism is just $1$-alternating nondeterminism.} guess integers, but the number of guesses they can make throughout the computation is bounded by a computable function of the parameter value of the input instance. We refer to this as \emph{parameter-bounded nondeterminism} (a term used similarly in \cite{cfg1}).

-The classes of the $\mathbf{A}$-Hierarchy, as well as $\mathbf{AW[*]}$, are obtained by further restricting the (alternating) nondeterminism of the machines to \emph{tail-nondeterminism}, meaning that the machines can only make nondeterministic guesses among the last $h(k)$ steps of a computation, where $h$ is a computable function and $k$ is the parameter. 

-Finally, the classes of the $\mathbf{W}$-Hierarchy are characterized by tail-nondeterministic machines which are not allowed to access the guessed integers directly (they can make nondeterministic decisions based on them, but not use them in arithmetic operations).

The main reason why the characterizations in \cite{cfg1}, \cite{cfg2}, and \cite{cfg3} were given in terms of RAMs, rather than Turing machines (TMs), is that a TM may need to traverse the entire used portion of its tape in order to read a particular bit, so a tail-nondeterministic TM would not be able to make use of its entire memory during the nondeterministic phase of the computation. The classes $\mathbf{W[P]}$ and $\mathbf{AW[P]}$ also have characterizations in terms of TMs with restricted nondeterminism \cite{cfg1}, but we consistently use random access machines throughout this work.

The machine characterizations of some of the above-mentioned classes can be rewritten in such a way that they strongly resemble definitions of some familiar classes from classical complexity. For example, $\mathbf{A}[1]$ can be defined as the class of parameterized problems that are decided by tail-nondeterministic RAMs in \emph{fpt-time}, which at least formally looks like the definition of $\mathbf{NP}$. Similarly, $\mathbf{W[P]}$ can also be defined in a way that is similar to $\mathbf{NP}$ (using parameter-bounded nondeterminism), the levels of the $\mathbf{A}$-Hierarchy have characterizations that match the definitions of the $\Sigma$-levels of the Polynomial Hierarchy, and $\mathbf{AW[P]}$ and $\mathbf{AW[*]}$ both correspond to $\mathbf{AP}$ (the class of problems that are decidable in alternating polynomial-time). Given the similar definitions, it seems reasonable to expect that parameterized complexity classes also inherit some properties from their classical counterparts. On the other hand, replacing the machine model in a definition is a significant change, so it is by no means obvious which theorems will still hold for a parameterized version of a complexity class.

Our goal in this paper is to show that having machine-based characterizations of parameterized complexity classes opens up a largely unexplored, but possibly very fruitful, path toward understanding parameterized intractability. To that end we extend the work of Chen, Flum, and Grohe \cite{cfg1, cfg2, cfg3} in two directions: relativization and interactive proofs. The key insight is that parameterized versions of these two concepts can be defined in such a way that some important classical theorems can be recovered in this setting. The proofs of our theorems follow along the same lines as their classical counterparts, with only some technical obstacles to be overcome, but it is a remarkable fact that parameterized versions of these proofs can be made to work at all: For example, it is not a priori clear whether parameterized oracle computation can be even in principle defined in a way that makes the $\mathbf{A}$-Hierarchy have an oracle-based characterization that is similar to that of $\mathbf{PH}$. We show, among other things, that this is indeed the case, and furthermore, that the restrictions that must be placed on the access to the oracle in order to obtain this result are quite natural (at least, in the context of the machine characterization of $\mathbf{A}[1]$ from \cite{cfg3}).

\subsection{Our results}
\textbf{Parameterized relativization.} Theorems involving oracles have been given before in Parameterized Complexity, but it is almost always Turing machines that are endowed with access to an oracle (see, for example, \cite{muell}). In order to relativize the hard parameterized complexity classes for which machine characterizations are known, we define oracle RAMs with the different forms of restricted nondeterminism mentioned above. It turns out that in order for oracle access and nondeterminism to interact in a useful way, both of these features must, roughly speaking, have the same restrictions (tail-nondeterministic machines should have tail-restricted oracle access, etc.).\footnote{Placing restrictions on the access to an oracle is a fairly common practice even in classical complexity. For example, the oracle tape of a $\mathbf{LOGSPACE}$-machine is write-only, in order to allow the machine to make polynomial-sized queries while preventing it from using the tape for computations that avoid the space restriction. Another example can be found in \cite{aw08}, where, in order to prove that the statement $\mathbf{NEXP}\subset\mathbf{MIP}$ \emph{algebrizes}, the authors restrict machines that run in exponential time so that they can only make poly-sized oracle queries.} We show that these restrictions lead to a natural type of oracle access for each type of machine, by proving parameterized versions of two fundamental results from classical complexity, both for the tail-nondeterministic and the parameter-bounded version of nondeterministic RAMs.

First, we give a new characterization of the classes of the $\mathbf{A}$-Hierarchy, in terms of oracle machines (resembling the oracle characterization of the levels of the Polynomial Hierarchy (see \cite{ab}, Section 5.5)), by proving that
\begin{displaymath}
\forall t\geq 1:\mathbf{A}[1]^{O_t}=\mathbf{A}[t+1],
\end{displaymath}
but only for a \emph{specific} oracle $O_t$ that is complete for $\mathbf{A}[t]$ (Theorem \ref{ah_oracl_thm}). We also explain why tail-nondeterminism appears to be too weak to allow for this theorem to be proved for an arbitrary $\mathbf{A}[t]$-complete problem. The situation is much better when the nondeterminism is only parameter-bounded, and we have (Theorem \ref{wp_oracl_thm}) that
\begin{displaymath}
\forall t\geq 1:\mathbf{W[P]}^{\mathbf{\Sigma}^{[P]}_t}=\mathbf{\Sigma}^{[P]}_{t+1},
\end{displaymath}
where $\mathbf{\Sigma}^{[P]}_t$ ($t\geq 1$) are the $\Sigma$-levels of the analogue of the Polynomial Hierarchy for the machine model with parameter-bounded nondeterminism (so $\mathbf{\Sigma}^{[P]}_1=\mathbf{W[P]}$). We emphasize that both of these theorems seem to hold only if the oracle $\mathbf{A}[1]$- and $\mathbf{W[P]}$-machines have exactly the right restrictions placed on their oracle access, and even then, tail-nondeterminism causes a number of non-trivial technical issues (see the proof of Theorem \ref{ah_oracl_thm}).

Second, we recover a parameterized version of a well-known oracle separation result of Baker, Gill, and Solovay \cite{bakgs}, by showing (Theorem \ref{oraclesep_thm}) that there exist parameterized oracles $A$ and $B$ such that
\begin{displaymath}
\mathbf{FPT}^A=\mathbf{A}[1]^{A}\textrm{\ \ and\ \ }\mathbf{FPT}^B\neq \mathbf{A}[1]^{B}.\footnote{A similar theorem was proved much earlier by Downey and Fellows \cite{df92}, but based on a different notion of relativization (unavoidably, since the machine models for $\mathbf{A}[1]$ were only discovered a decade later \cite{cfg1}), and using heavily recursion-theoretic proof techniques.}
\end{displaymath}
It is worth noting that here the $\mathbf{FPT}$-machine may be given completely unrestricted access to the oracle $B$, whereas the $\mathbf{A}[1]$-machine only has tail-restricted access (which is the most restricted form of oracle access we consider), so in some sense this separation is stronger than expected. A similar theorem holds when replacing $\mathbf{A}[1]$ with $\mathbf{W[P]}$ (Theorem \ref{wp_oraclesep_thm}).

These results are, of course, only the first steps toward understanding relativization for parameterized complexity classes beyond $\mathbf{FPT}$. To illustrate the importance of investigating relativization in this setting, let us briefly consider the long-standing open problem of proving a parameterized version of Toda's Theorem \cite{toda}, which states that $\mathbf{PH}\subseteq \mathbf{P}^\mathbf{PP}$. It is not clear which parameterized classes would be involved in such a theorem, but, presumably, $\mathbf{P}$ would be replaced by $\mathbf{FPT}$, which can easily be described in terms of Turing machines, so it should be possible to at least state the theorem without further considerations about the type of oracle access being used. Furthermore, it could be argued that since only the larger of the two classes in the theorem statement is obtained via relativization, placing no restrictions on the access to the oracle can only make the inclusion easier to prove. However, both Toda's original proof \cite{toda} and Fortnow's simplified version of it \cite{fortnow} make heavy use of relativized versions of classes such as $\mathbf{BPP}$ and $\mathbf{PH}$, so following either one of these proofs would involve relativized versions of parameterized counterparts of such classes. Our Theorems \ref{ah_oracl_thm} and \ref{wp_oracl_thm} only deal with oracle access and alternating nondeterminism, but this already requires a careful balancing of the restrictions placed on both features. Toda's Theorem, on the other hand, involves an interplay between relativization, alternating nondeterminism, randomization, and counting complexity, so it seems unlikely that a parameterized version of it can be proved without a better understanding of parameterized relativization and its relation to other complexity-theoretic concepts.
\\\\
\textbf{Interactive proof systems for parameterized complexity classes.} The levels of the $\mathbf{A}$-Hierarchy were originally defined as fpt-closures of \emph{model checking} problems, where a relational structure $\mathcal{A}$ and a first-order formula $\phi$ without free variables are given, and the task is to decide whether $\mathcal{A}$ satisfies $\phi$. In \cite{cfg3}, model checking problems are used in a very interesting way in the proof of the machine characterization of the classes $\mathbf{A}[t]$: Specifically, a pair $(\mathcal{A},\phi)$ is used to encode the computation of a tail-nondeterministic RAM, in a way that is strongly reminiscent of how the computation of a nondeterministic TM is encoded as a quantified Boolean formula in the proof of the Cook-Levin Theorem (see \cite{ab}, Chap.\ 2). This suggests that by generalizing classical techniques that involve quantified Boolean formulas, it may be possible to apply them to parameterized complexity classes for which a model checking problem is complete. In Section \ref{ip_sec} we continue this line of thought by generalizing \emph{arithmetization} of quantified Boolean formulas (see \cite{ab}, Section 8.3) to pairs of relational structures and first-order formulas.

We also initiate the study of interactive proof systems in this setting. Using generalized arithmetization, we show that all problems in $\mathbf{AW[SAT]}$ have proof systems with a number of rounds depending only on the parameter value of the input instance (Theorem \ref{awsat_thm}). The goal (which, unfortunately, is not achieved here) is to precisely characterize either $\mathbf{AW[*]}$ or $\mathbf{AW[P]}$ in terms of IPs, as this would recover a parameterized version of the fact that $\mathbf{IP}=\mathbf{AP}$, even without a notion of space that corresponds to alternation in the parameterized setting. At the end of Section 4 we give a possible candidate for a characterization of $\mathbf{AW[*]}$.

\section{Preliminaries}
We refer to \cite{ab} and to \cite{fg}, respectively, for the necessary background in classical and Parameterized Complexity. By $\mathbb{N}$ we mean the set of non-negative integers, and by $\mathbb{N}^*$ the set of finite sequences of non-negative integers.
\subsection{Random access machines and parameterized complexity classes}

We give only a general overview of RAMs, and refer to Section 2.6 of \cite{papadim} for the details. A random access machine is specified by its \emph{program} (a finite sequence of instructions), which operates on an infinite sequence of \emph{standard registers}, $r_0,r_1,\ldots$, that contain integers. Instructions access registers either directly, by referencing their numbers, or indirectly, by taking the number of a register to be the current content of another register (in other words, the machine can access $r_{r_i}$, $i\in\mathbb{N}$, in constant time). We follow \cite{cfg1} in assuming that the registers store only non-negative integers. Except instructions that copy the contents of one register to another, a RAM also has conditional and unconditional jump instructions, as well as instructions which perform the operations addition, subtraction, and integer division by 2 (these suffice to efficiently perform all arithmetic operations on signed integers). The input of a RAM is a finite sequence of non-negative integers, each stored in a separate register, and we define the problems solved by such machines accordingly.

\begin{defin}
A \emph{parameterized problem} $Q$ is a subset of $\mathbb{N}^{*}\times\mathbb{N}$. When dealing with the problem of deciding whether $(x,k)\in \mathbb{N}^{*}\times\mathbb{N}$ is an element of $Q$, $(x,k)$ is referred to as an \emph{instance}; the second element of such a pair is called the \emph{parameter}. We assume for simplicity that $k\leq |x|$ holds for all elements of the problems we work with (since instances with $k>|x|$ are trivial).
\end{defin}

\begin{rem}\label{input_rem}
When an instance of a parameterized problem is given as input to a RAM, we assume that the parameter is given in unary encoding, meaning that if the parameter value is $k\in\mathbb{N}$, then $k$ registers, each containing the value $1$, are used to encode the parameter value. The size of $x$, the main part of the input, is taken as the sum of the sizes of the binary encodings of the integers that make up $x$. A RAM can therefore efficiently convert between a reasonable encoding using integers, and any reasonable encoding using a finite alphabet.
\end{rem}

\begin{defin}
A random access machine $\mathbb{M}$ is \emph{parameter-restricted} if there is a computable function $f$ and a polynomial function $p$, such that on any input $(x,k)$:

-$\mathbb{M}$ terminates after executing at most $f(k)p(|x|)$ instructions;

-throughout any computation, the registers contain only numbers that are $\leq f(k)p(|x|)$.
\end{defin}

The above definition replaces the ``polynomial-time'' restriction on the running time in the classical setting, and is similar to the definition of ``$\kappa$-restricted'' in Chap.\ 6 of \cite{fg}. Note that the second condition is a bound on the numbers stored in the registers, not on the number of bits that would be needed for the binary encoding of these numbers.

The next definition is easily seen to be equivalent to the usual definition of the class $\mathbf{FPT}$ \cite{fg}.

\begin{defin}
We define $\mathbf{FPT}$ as the class of parameterized problems that are decidable by parameter-restricted (deterministic) RAMs.
\end{defin}

An \emph{alternating random access machine (ARAM)} is a RAM with additional \emph{existential} and \emph{universal guess instructions}, EXISTS and FORALL, both of which place a nondeterministically chosen integer from the interval $[0,r_0]$ into $r_0$ (the difference between the two instructions is in how the acceptance of the input is defined). In the case of parameter-restricted machines, we may assume that the upper end of the range of each nondeterministic guess is the largest number that the machine can store in its registers, given the input, because the machine can first guess a number in the maximum range, and then trim the result by computing the remainder of a division by the size of the intended range. For ARAMs, the notions of computation (on an input), configuration, computation path, $t$-alternation, and acceptance/rejection of an input are defined in the standard way (see \cite{fg}, section 8.1, pp. 168-170). Following \cite{cfg3}, we mean by ``$t$-alternating'' that the first guess instruction is existential.

We give the definitions of some complexity classes in terms of nondeterministic RAMs. These are not the original definitions, but characterizations proved in \cite{cfg1} and \cite{cfg3}.

\begin{defin}
A parameterized problem $Q$ is in $\mathbf{AW[P]}$ [in $\mathbf{W[P]}$] if it is decided by an ARAM [a $1$-alternating ARAM] $\mathbb{A}$ which, for some computable function $h$, on any input $(x,k)$, executes at most $h(k)$ nondeterministic instructions on any computation path.
\end{defin}

\begin{defin}
An ARAM $\mathbb{A}$ is \emph{tail-nondeterministic} if there is a computable function $g$ such that, on any input $(x,k)$, $\mathbb{A}$ executes nondeterministic instructions only among the last $g(k)$ steps of any computation path. For every $t\geq 1$, $\mathbf{A}[t]$ denotes the class of parameterized problems that are decidable by parameter-restricted tail-nondeterministic $t$-alternating ARAMs. $\mathbf{AW[*]}$ denotes the class of parameterized problems that are decidable by parameter-restricted tail-nondeterministic ARAMs.
\end{defin}

An \emph{oracle (A)RAM} or \emph{(A)RAM with access to an oracle} is a machine with an additional set of \emph{oracle registers} that store non-negative integers, as well as instructions that copy the contents of $r_0$ to a specified oracle register and vice-versa, and a QUERY instruction, which queries the oracle with the contents of the oracle registers, and causes the register $r_0$ to contain the values $1$ or $0$ (representing the oracle's answer). Note that we only work with oracles that decide parameterized problems, and that the parameter of a query instance must be encoded in unary (see Remark \ref{input_rem}). Most previous results involving oracles in Parameterized Complexity place the following restriction on oracle machines (see, for example, \cite{muell}). We will consider additional restrictions to oracle access in the next section.

\begin{defin}
An oracle (A)RAM $\mathbb{A}$ has \emph{balanced} access to an oracle if there is a computable function $g$ such that, on input $(x,k)$, any query $(y,k')$ made to the oracle, on any computation path, satisfies $k'\leq g(k)$.
\end{defin}

\subsection{Relational structures and first-order formulas}
A \emph{relational vocabulary} $\tau$ is a set of pairs of symbols and positive integers, called \emph{relational symbols} and \emph{arities}, respectively. A \emph{relational structure $\mathcal{A}$ with vocabulary $\tau$} is a set containing: a set $A$, called the \emph{universe of $\mathcal{A}$}, and for each pair $(s,r)\in \tau$, a relation $R^s\subseteq A^r$. We only use relational structures with finite universes and finite vocabularies, so we always assume that $A=\{0,\ldots,n\}$, for some $n\in\mathbb{N}$. A \emph{first-order formula $\phi$ with vocabulary $\tau$} is constructed in the same way as a quantified Boolean formula, except that the \emph{atomic formulas} are not variables, but expressions of the form $x_1=x_2$ or $R^s x_1\ldots x_r$, where $x_1,x_2,\ldots,x_r$ are variables and $(s,r)\in\tau$. 

Whenever a pair $(\mathcal{A},\phi)$ is given, it is assumed implicitly that $\mathcal{A}$ and $\phi$ share the same relational vocabulary. We say that $\mathcal{A}$ \emph{satisfies $\phi$} if $\phi$ is true when all atomic formulas are evaluated based on the relations in $\mathcal{A}$ and all variables are taken as ranging over $A$.

We define some important classes of first-order formulas with relational vocabularies. For every $t\in\mathbb{N}$, let $\Sigma_t$ be the set of all first-order formulas of the form
\begin{displaymath}
\exists x_{1,1}\ldots\exists x_{1,k_1}\forall x_{2,1}\ldots \forall x_{2,k_2}\ldots \ldots Q x_{t,1}\ldots Q x_{t,k_t}:\psi(x_1,\ldots,x_t),
\end{displaymath}
where $\psi(x_1,\ldots,x_t)$ is a quantifier-free formula ($Q$ means $\exists$ if $t$ is odd, $\forall$ if $t$ is even). For all $t,r\in\mathbb{N}$, let $\Sigma_t[r]$ be the set of all $\Sigma_t$-formulas with vocabularies in which all arities are $\leq r$. Finally, let $\textrm{PNF}$ be the set of all first-order formulas in \emph{prenex normal form}, meaning that they are of the form $Q_1 x_1\ldots Q_t x_t:\psi(x_1,\ldots,x_t)$, where $\psi(x_1,\ldots,x_t)$ is a quantifier-free formula and $Q_1,\ldots,Q_t\in\{\exists,\forall\}$.

For certain classes of formulas $F$, the following parameterized \emph{model checking} problems are complete for various important complexity classes.
\begin{center}
\fbox{
\begin{minipage}{12cm}
$p\textsc{-MC}(F)$\\
\begin{tabular}{ r l }
Input: & \parbox[t]{9cm}{$(\mathcal{A},\phi)$, where $\mathcal{A}$ is a relational structure, $\phi\in F$.}\\
Parameter: & $|\phi|.$\\
Problem: & Decide whether $\mathcal{A}$ satisfies $\phi$.
\end{tabular}
\end{minipage}}
\end{center}

\begin{center}
\fbox{
\begin{minipage}{12cm}
$p\textit{-var}\textsc{-MC}(F)$\\
\begin{tabular}{ r l }
Input: & \parbox[t]{9cm}{$(\mathcal{A},\phi)$, where $\mathcal{A}$ is a relational structure, $\phi\in F$.}\\
Parameter: & The number of variables in $\phi$.\\
Problem: & Decide whether $\mathcal{A}$ satisfies $\phi$.
\end{tabular}
\end{minipage}}
\end{center}

\begin{rem}\label{rep_rem}
A relational structure can be represented by listing the elements of its universe, followed by the tuples in each relation. However, for a RAM to check whether some tuple $(a_1,\ldots,a_r)$ is an element of some $r$-ary relation $R^s$ may then take a number of steps that depends on $\|\mathcal{A}\|:=|A|+|\tau|+\sum_{(s,r)\in \tau}|R^s|\cdot r$ (even if the elements of each relation are listed in lexicographic order, and binary search is used). To avoid this we will assume, whenever $\mathcal{A}$ contains only relations of arity at most some fixed number $l$, that each $r$-ary relation ($r\leq l$) is stored as an $|A|^r$-size array of ones and zeroes, each number representing whether or not some element of $A^r$ is a member of the relation. Furthermore, we will assume that the location of every such array is stored in a look-up table. This way, checking whether $(a_1,\ldots,a_r)\in R^s$ only takes a \emph{constant} number of operations for a RAM, at the cost of increasing the size of the representation of $\mathcal{A}$ in memory to $O(\poly(\|\mathcal{A}\|))$ (since $l$ is constant). This also means that adding and removing elements requires only constant time.
\end{rem}

\begin{defin}
Let $Q$ and $Q'$ be parameterized problems. An algorithm $\mathbb{R}$ is an \emph{fpt-reduction from $Q$ to $Q'$} if there exist computable functions $f$ and $g$, and a polynomial function $p$, such that for any instance $(x,k)$ of $Q$ we have a) $(y,k'):=\mathbb{R}(x,k)\in Q'$ if and only if $(x,k)\in Q$; b) $\mathbb{R}$ runs in time $f(k)p(|x|)$; and c) $k'\leq h(k)$.
\end{defin}

For any parameterized problem $Q$, we denote by $[Q]^{\textrm{fpt}}$ the set of parameterized problems that are $\leq^{\textrm{fpt}}Q$, meaning fpt-reducible to $Q$.

\begin{fact}[\cite{fg2001, cfg1},\cite{abrah}]
For every $t\in\mathbb{N}$, $\mathbf{A}[t]=[p\textsc{-MC}(\Sigma_t)]^{\textrm{fpt}}=[p\textsc{-MC}(\Sigma_t[3])]^{\textrm{fpt}}$.

$\mathbf{AW[SAT]}=[p\textit{-var}\textsc{-MC}(\textrm{PNF})]^{\textrm{fpt}}$.
\end{fact}

\begin{rem}\label{reduction_rem}
In the proof of their machine-based characterization of $\mathbf{A}[t]$, Chen, Flum, and Grohe \cite{cfg3} show how the parameter-restricted computation of a $t$-alternating tail-nondeterministic RAM can be encoded as a pair $(\mathcal{A},\phi)$. We refer the interested reader to \cite{cfg3} for the details, and recall only some facts about this reduction that we use here. Let $f(k)p(|x|)$ be an upper bound on the running time, the largest number of a register used, and the largest integer stored during the computation of the machine $\mathbb{A}$ on input $(x,k)$. The relational structure $\mathcal{A}$ has universe $\{0,\ldots,f(k)p(|x|)\}$ and contains relations representing the instructions of $\mathbb{A}$'s program and the contents of the accessed registers at the end of the deterministic part of the computation (a binary relation $Reg$ is defined so that $(y,z)\in Reg$ if and only if $r_y=z$ right before the first nondeterministic instruction is executed). All relations in $\mathcal{A}$ have arity $\leq 3$. The first-order formula $\phi$ has the same vocabulary as $\mathcal{A}$ and encodes the nondeterministic computation of $\mathbb{A}$ (the last $h(k)$ steps). The formula is constructed in such a way that changes to the contents of the registers are kept track of, and access to the contents of the registers at the start of the nondeterministic computation are encoded using the relation $Reg$. A close look at the construction in \cite{cfg3} reveals that computing the relational structure $\mathcal{A}$ requires knowledge of $\mathbb{A}$ and of the input $(x,k)$, but that computing the formula $\phi$ only requires knowledge of $k$, $\mathbb{A}$, and the number of the first nondeterministic instruction that is executed on input $(x,k)$ (all of these being independent of $|x|$).
\end{rem}

\section{Parameterized relativization}\label{rel_sec}
The guiding principle in our approach to defining nondeterministic oracle RAMs will be that all of the special resources of a machine (nondeterminism, oracle queries, random guesses -- everything beyond the basic deterministic operations) should be restricted in the same way, in order for these resources to interact well with each other.

\begin{defin}
An oracle (A)RAM $\mathbb{A}$ has \emph{parameter-bounded} access to an oracle if it has balanced access to the oracle, and there is a computable function $h$ such that, on input $(x,k)$, $\mathbb{A}$ makes at most $h(k)$ queries to the oracle on any computation path. $\mathbb{A}$ is said to have \emph{tail-restricted} access to an oracle if it has balanced access to the oracle, and there is a computable function $h$ such that, on input $(x,k)$, $\mathbb{A}$ makes queries to the oracle only among the last $h(k)$ steps of any computation path.
\end{defin}

Because we will use different kinds of oracle machines, and the exponent notation for the relativization of a complexity class is difficult to customize, we will also use the (older) parenthesis notation: If $C$ is a complexity class that is characterized by machines, we denote by $C(O)$ the class characterized by oracle machines of the same type as the ones characterizing $C$, with unrestricted access to the oracle $O$. Similarly, $C(O)_{bal}$ denotes the class defined by oracle $C$-machines with \emph{balanced} access to the parameterized oracle, $C(O)_{para}$ denotes the class defined by oracle $C$-machines with \emph{parameter-bounded} access to the oracle, and $C(O)_{tail}$ denotes the class defined by tail-nondeterministic oracle machines with the same restrictions as the machines that define $C$. The exponent notation is only used when the type of oracle access is the ``natural'' one for the type of machine being considered (so $\mathbf{A}[1]^O=\mathbf{A}[1](O)_{tail}$ and $\mathbf{W[P]}^O=\mathbf{W[P]}(O)_{para}$). For $\mathbf{FPT}$ we always specify the type of oracle access.
\\\\
\emph{Relativization results for tail-nondeterministic random access machines.}

We give an informal overview of the proof that $\mathbf{A}[1]^{p\textsc{-MC}(\Sigma_t[3])}=\mathbf{A}[t+1]$, to highlight the role played by the choice of the oracle and by the restrictions made to the tail-nondeterministic oracle machines (for a comparison with the proof that $\mathbf{NP}^{\Sigma_i\textsc{Sat}}=\mathbf{\Sigma}^P_{i+1}$, see \cite{ab}, Section 5.5).

For the ``$\supseteq$''-inclusion, we have that an $\mathbf{A}[1]$-machine with a $p\textsc{-MC}(\Sigma_t[3])$-oracle (which is complete for $\mathbf{A}[t]$) can first deterministically simulate the deterministic part of the computation of an $\mathbf{A}[t+1]$-machine on input $(x,k)$. The oracle $\mathbf{A}[1]$-machine then enters the nondeterministic phase of its computation and uses its own nondeterministic guesses to simulate the first block of existential guesses of the simulated machine (until a universal instruction is encountered). The computation of the $\mathbf{A}[t+1]$-machine from this point onward (which starts with a universal guess instruction and has $\leq t-1$ alternations) can be encoded as an instance $((\mathcal{A},\phi),|\phi|)$ of $p\textsc{-MC}(\Sigma_t[3])$ (see Remark \ref{reduction_rem}), but the size of $\mathcal{A}$ depends on $|x|$. Therefore, $\mathcal{A}$ must (for the most part) be computed by the oracle $\mathbf{A}[1]$-machine and written to the oracle registers ahead of time, during the deterministic phase of the computation, with only the formula $\phi$ left to be computed during the nondeterministic phase. This is why it is necessary to allow tail-nondeterministic oracle machines access to their oracle registers throughout the entire computation.

For the reverse inclusion, we have that an $\mathbf{A}[t+1]$-machine can simulate an oracle $\mathbf{A}[1]$-machine on input $(x,k)$, by first simulating the deterministic part of the computation deterministically, and then using $(t+1)$-alternating nondeterminism to simulate both the oracle $\mathbf{A}[1]$-machine's existential guesses, as well as all of the $p\textsc{-MC}(\Sigma_t[3])$-queries (this is accomplished in the same way as in the classical proof). In order to evaluate the queried instances, however, the $\mathbf{A}[t+1]$-machine's computation must be in its nondeterministic phase, so it is essential that:

- the simulated oracle machine can not make queries outside of the last $h(k)$ steps of its computation, for some computable function $h$;

- the size of the formulas in the queried instances is $\leq g(k)$, for some computable function $g$ (balanced oracle access);

- the quantifier-free part of a formula can be evaluated efficiently (relational structures must be encoded in such a way that expressions involving relations can be evaluated by a RAM in time independent of the size of the relational structure -- see Remark \ref{rep_rem}).

\begin{theor}\label{ah_oracl_thm} For every $t\geq 1$, $\mathbf{A}[1]^{p\textsc{-MC}(\Sigma_t[3])}=\mathbf{A}[t+1]$.
\end{theor}

\begin{proof}
As mentioned in Remark \ref{reduction_rem}, the computation of an $\mathbf{A}[t]$-machine can be encoded as an instance of $p\textsc{-MC}(\Sigma_t[3])$, and here we make extensive use of this reduction. Note, however, that in \cite{cfg3}, relational structures and FO formulas can contain \emph{constant symbols}, which are interpreted as representing fixed values from the universe of the relational structure, when evaluating a formula. Within the scope of this proof we will therefore also allow instances of $p\textsc{-MC}(\Sigma_t[3])$ to contain such constant symbols, but this does not change the fact that $[p\textsc{-MC}(\Sigma_t[3])]^{\textrm{fpt}}=\mathbf{A}[t]$ for all $t\geq 1$ \cite{cfg3}, nor does it cause any complications in the proof.

``$\supseteq$'': Let $Q$ be a parameterized language in $\mathbf{A}[t+1]$. Then, for some computable functions $f,h$, and a polynomial function $p$, there is a $(t+1)$-alternating ARAM $\mathbb{A}$ which, on any input $(x,k)$, runs in time $f(k)p(|x|)$, with nondeterministic instructions only among the last $h(k)$ steps, and accepts if and only if $(x,k)\in Q$.

We describe an $\mathbf{A}[1]$-machine $\mathbb{A}_1$ with oracle access to $p$-\textsc{MC}($\Sigma_t[3]$) which decides $Q$. Let $\mathbb{A}'$ be an $\mathbf{A}[t+1]$-machine obtained from $\mathbb{A}$ by replacing its HALT instruction with a sequence of instructions that flip the output before halting. (Note that we may assume that the program of every machine has a single HALT instruction, because it is trivial to modify a program in such a way that it has only one such instruction.) Let $\mathbb{A}''$ be an $\mathbf{A}[t+1]$-machine obtained from $\mathbb{A}'$ by replacing all EXISTS instructions with FORALL instructions, and vice-versa. Let $l$ be the (constant) number of additional instructions in the program of $\mathbb{A}'$. On input $(x,k)$, $\mathbb{A}_1$ does the following:
\begin{itemize}
\item[1.] $\mathbb{A}_1$ simulates $\mathbb{A}'$ on input $(x,k)$ until the first non-deterministic instruction is about to be executed by $\mathbb{A}'$. (If the simulation terminates before a non-deterministic instruction is executed, $\mathbb{A}_1$ halts as well, with opposite outcome.) After this part, the contents of the standard registers of $\mathbb{A}'$ right before it executes its first nondeterministic instruction, are accessible to $\mathbb{A}_1$.
\item[2.] $\mathbb{A}_1$ computes the relational structure $\mathcal{A}$, with universe $\{0,\ldots, f(k)p(|x|)\}$, containing:

-constant symbols for all instruction numbers and register numbers referred to by instructions of $\mathbb{A}'$'s program;

-relations encoding the standard instructions of an ARAM, restricted to the universe of $\mathcal{A}$;

-a relation $Reg$, encoding the contents of $\mathbb{A}'$'s standard registers right before it executes its first nondeterministic instruction.

A representation of this relational structure is stored in the oracle registers of $\mathbb{A}_1$.
\item[3.] $\mathbb{A}_1$ now enters the nondeterministic phase of its computation. It guesses (with existential quantifier) $h(k)$ integers between $0$ and $f(k)p(|x|)$, and stores them in the standard registers.
\item[4.] Using the numbers guessed in part 3 to simulate EXISTS instructions, $\mathbb{A}_1$ continues simulating $\mathbb{A}'$ until the first FORALL instruction is encountered. (If the simulation terminates before a universal guess instruction is encountered, $\mathbb{A}_1$ halts with the opposite outcome.) Throughout this part, $\mathbb{A}_1$ also modifies the relation $Reg$ in $\mathcal{A}$ to reflect the changes to the contents of the registers of $\mathbb{A}'$ (at most $h(k)+l$ changes, because the simulated computation is in its nondeterministic phase).
\item[5.] Let $c$ be the number of the FORALL instruction encountered in part 4, and let $d$ be the number of steps of $\mathbb{A}'$'s computation that were simulated during part 4. Now $\mathbb{A}_1$ computes a $\Sigma_{t}$-formula $\phi$, which encodes a $t$-alternating computation of $\mathbb{A}''$ of at most $h(k)+l-d$ steps, starting at instruction $c$ (see Remark \ref{reduction_rem}). This formula is stored in the oracle registers, after the representation of the structure $\mathcal{A}$.
\item[6.] After part 5, the oracle registers together contain an instance of $p$-\textsc{MC}($\Sigma_t[3]$). Finally, $\mathbb{A}_1$ queries the $p$-\textsc{MC}($\Sigma_t[3]$)-oracle and accepts if the oracle's answer is negative, otherwise it rejects.
\end{itemize}

The computations performed in parts 3-6 require some number of steps that is a computable function of $k$, and independent of $|x|$ (this is trivial for parts 3, 4 and 6, and holds for part 5 because the formula encodes a computation of $O(h(k))$ steps, and therefore has size $h'(h(k))$, for some computable function $h'$). To see why $\mathbb{A}_1$ accepts on input $(x,k)$ if and only if $(x,k)\in Q$, note that, on a given input, a computation of $\mathbb{A}$ up to the first FORALL instruction is essentially identical to a computation of $\mathbb{A}''$ up to the first EXISTS instruction (unless the computation terminates earlier, in which case only the outcome differs). Since $\mathbb{A}''$ has both the guess instruction types and the outcome reversed, the oracle will answer `yes' if and only if the computation of $\mathbb{A}$ after the first FORALL instruction does not accept the input.

``$\subseteq$'': Let $Q$ be a parameterized problem in $\mathbf{A}[1]^{p\textsc{-MC}(\Sigma_t[3])}$. Then, for some computable functions $f,h$, and a polynomial function $p$, there is an ARAM $\mathbb{A}_1$ with tail-restricted access to a $p\textsc{-MC}(\Sigma_t[3])$-oracle, which, on any input $(x,k)$, runs in time $f(k)p(|x|)$, with nondeterministic and oracle query instructions only among the last $h(k)$ steps, and accepts the input if and only if $(x,k)\in Q$. Furthermore, since $\mathbb{A}_1$'s oracle access is balanced (due to being tail-restricted), there is a computable function $g$ such that, on input $(x,k)$, any query made to the oracle has parameter value $k'\leq g(k)$.

We describe an $\mathbf{A}[t+1]$-machine $\mathbb{A}$ that decides $Q$. On input $(x,k)$, $\mathbb{A}$ does the following:
\begin{itemize}
\item[1.] $\mathbb{A}$ simulates $\mathbb{A}_1$ until either the first EXISTS instruction or the first QUERY instruction is encountered. Throughout this part, the contents of both the standard and the oracle registers of $\mathbb{A}_1$ are stored and maintained separately in the standard registers of $\mathbb{A}$.
\item[2.] $\mathbb{A}$ computes a relational structure $\mathcal{A}$, encoding the instruction set of $\mathbb{A}_1$ and the contents of $\mathbb{A}_1$'s standard and oracle registers right before it executes either its first existential or its first oracle query instruction.
\item[3.] $\mathbb{A}$ makes a number of guesses with existential quantifiers:\\
-$g_1,\ldots,g_{h(k)}\in\{0,\ldots,f(k)p(|x|)\}$ ($\mathbb{A}$'s guesses for all of the existential guesses $\mathbb{A}_1$ will make);\\
-$a_1,\ldots,a_{h(k)}\in\{0,1\}$ ($\mathbb{A}$'s guesses for the answers to $\mathbb{A}_1$'s oracle queries).\vspace{0.2cm}\\
Next, $\mathbb{A}$ makes the following guesses with alternating quantifiers:\\
-$u_{i,j}\in \{0,\ldots,f(k)p(|x|)\}^{g(k)}$, with $i\in[t+1]$ and $j\in[h(k)]$, where the strings $u_{1,j}$ (for all $j$) are guessed as one block with an existential quantifier, the strings $u_{2,j}$ are guessed as a block with a universal quantifier, and so on (existential quantifier for odd $i$, universal for even). (These strings form $\mathbb{A}$'s guesses for witness strings for $\mathbb{A}_1$'s queries.)
\item[4.] Now $\mathbb{A}$ continues simulating $\mathbb{A}_1$ as follows:\\
-on the $i$-th (existential) guess of $\mathbb{A}_1$, $\mathbb{A}$ uses $g_i$ instead of the guess value;\\
-on the $j$-th query of $\mathbb{A}_1$ ($j\in[h(k)]$), let 
\begin{displaymath}
\phi_j=\exists x_{1,1},\ldots,\exists x_{1,m_1}\forall x_{2,1},\ldots,\forall x_{2,m_2}\ldots Qx_{t,1},\ldots,Qx_{t,m_t}\psi(x_{1,1},\ldots,x_{t,m_t})
\end{displaymath}
be the $t$-alternating formula of the query instance, where $\psi$ is quantifier-free and $m_1+\ldots+m_t\leq g(k)$. Let $\mathcal{A}_j$ be the relational structure of the $j$-th query instance, and let $U_j$ be the size of the universe of $\mathcal{A}_j$ (so we may assume that $A_j=\{0,\ldots,U_j-1\}$). If $a_j=1$ ($\mathbb{A}$ guessed that the oracle answers 'yes'), $\mathbb{A}$ evaluates $\phi_j$ by setting $x_{k,l}$ to $(u_{k,j,l}\mod U_j)$, for each $k\in[t],l\in[m_k]$. If $a_j=0$, $\mathbb{A}$ evaluates $\phi_j$ by setting $x_{k,l}$ to $(u_{k+1,j,l}\mod U_j)$, for each $k\in[t],l\in[m_k]$. If the outcome of the evaluation does not match $a_j$, $\mathbb{A}$ rejects, otherwise it continues the computation.

\item[5.] If the simulation terminates, accept if and only if $\mathbb{A}_1$ terminated on an accepting configuration.
\end{itemize}
$\mathbb{A}_1$ accepts an input if and only if, for some sequence of existential guesses, it terminates on an accepting configuration. If such a sequence of guesses exists, then $\mathbb{A}$ can existentially guess it (via $g_1,\ldots,g_{h(k)}$), as well as the correct oracle answers on this computation path (via $a_1,\ldots,a_{h(k)}$). Now, if a queried instance $((\mathcal{A}_j,\phi_j),|\phi_j|)$ is a 'yes'-instance (and hence $a_j=1$), then
\begin{align*}
\exists u_{1,j,1},\ldots,\exists u_{1,j,m_1}\forall u_{2,j,1},\ldots,\forall u_{2,j,m_2}\ldots Qu_{t,j,1},\ldots,Qu_{t,j,m_t}:&\\
\psi(u_{1,j,1}\textrm{\ mod\ }U_j,\ldots,u_{t,j,m_t}\textrm{\ mod\ }U_j)=1&,
\end{align*}
otherwise
\begin{align*}
\forall u_{2,j,1},\ldots,\forall u_{2,j,m_1}\exists u_{3,j,1},\ldots,\exists u_{3,j,m_2}\ldots Qu_{t+1,j,1},\ldots,Qu_{t+1,j,m_t}:&\\
\psi(u_{2,j,1}\textrm{\ mod\ }U_j,\ldots,u_{t+1,j,m_t}\textrm{\ mod\ }U_j)=0&.
\end{align*}
In other words, once $\mathbb{A}$ has correctly guessed the answers to all oracle queries, it will also correctly evaluate each query instance using alternating nondeterminism, and produce the same outcome as the simulated machine.
\end{proof}

Since, for every $t\geq 1$, the problem used as an oracle in Theorem \ref{ah_oracl_thm} is complete for $\mathbf{A}[t]$, it would be tempting to now state that $\mathbf{A}[1]^{\mathbf{A}[t]}=\mathbf{A}[t+1]$, because this would imply a ``collapse theorem'' for this hierarchy, namely that $\forall t\geq 1:\mathbf{A}[t]=\mathbf{A}[t+1]\Rightarrow (\forall t'\geq t:\mathbf{A}[t]=\mathbf{A}[t'])$. Unfortunately, tail-nondeterminism appears to be too weak for such a collapse theorem to be proved in this fashion. In fact, it is not even certain whether $\mathbf{A}[1]^{\mathbf{FPT}}\subseteq \mathbf{A}[2]$: This is because an $\mathbf{A}[2]$-machine trying to simulate an $\mathbf{A}[1]$-machine that has oracle access to some non-trivial problem in $\mathbf{FPT}$, on some input $(x,k)$, may have to enter the nondeterministic phase of its computation before it even knows the instance to be queried (the simulated machine may write a large instance to its oracle registers, and then nondeterministically make some changes to it before querying the oracle). The size of this instance may depend on $|x|$, and although it can be decided in fpt-time, it may not be possible to decide it in time $h(k)$, for some computable function $h$, even with $2$-alternating nondeterminism. Thus, the property of $p\textsc{-MC}(\Sigma_t[3])$ that, with the right encoding, an instance $((\mathcal{A},\phi),|\phi|)$ can be decided by a $t$-alternating tail-nondeterministic ARAM in time depending computably only on $|\phi|$, appears to have been crucial for our oracle characterization of the $\mathbf{A}$-Hierarchy.
\\\\
The next theorem is the parameterized analogue of a famous classical result of Baker, Gill, and Solovay \cite{bakgs}. The construction of a parameterized oracle $B$ relative to which $\mathbf{FPT}$ and $\mathbf{A}[1]$ differ, is done via diagonalization and uses similar ideas as the classical proof in \cite{bakgs}, but with two noteworthy differences:

First, when diagonalizing against all $\mathbf{FPT}$-machines, we can not computably list all such machines, because the $f(k)$-term in their running times can be any computable function (it is not even possible to computably list a sequence of computable functions such that every computable function is asymptotically dominated by some function in the list). We must therefore proceed more carefully with the construction in order to obtain an oracle that is computable.

Second, when running each RAM on larger and larger inputs for an increasing number of steps while constructing the oracle, we are free to increase \emph{both} the size of the main part of the input \emph{and} the parameter value. Having this additional dimension of the input works in our favor, and allows us to ``kill'' the $f(k)$-term in the running time of any $\mathbf{FPT}$-machine by increasing $|x|$ so that $|x|>f(k)$, at which point we can treat $f(k)|x|^c$ as a polynomial in $|x|$.

\begin{theor}\label{oraclesep_thm}
There exist parameterized oracles $A$ and $B$ such that 
\begin{align*}
\mathbf{FPT}(A)_{tail}=\mathbf{A}[1]^{A}\textrm{\ and\ }\mathbf{A}[1]^{B}\setminus \mathbf{FPT}(B)\neq \emptyset.
\end{align*}
\end{theor}

\begin{proof}
Consider the following parameterized problem:
\begin{center}
\fbox{
\begin{minipage}{12cm}
\textsc{XP-RAM-Computation}\\
\begin{tabular}{ r l }
Input: & A RAM $\mathbb{A}$, an input $(x,k)$, and $n\in\mathbb{N}$ in unary. \\
Parameter: & $k'\in\mathbb{N}$.\\
Problem: & \parbox[t]{9cm}{Decide whether $\mathbb{A}$ accepts the input $(x,k)$ in at most $n^{k'}+k'$ steps.}
\end{tabular}
\end{minipage}}
\end{center}

Let $A=\textsc{XP-RAM-Computation}$. Evidently, $\mathbf{FPT}(A)_{tail}\subseteq\mathbf{A}[1]^{A}$. Let $\mathbb{A}$ be an oracle $\mathbf{A}[1]$-machine, and let $f,g,h$ be computable functions, and $p$ a polynomial function, such that on input $(x,k)$, $\mathbb{A}$ runs in time $f(k)p(|x|)$, executes nondeterministic or oracle query instructions only among the last $h(k)$ steps of any computation, and queries the oracle only for instances with parameter value $\leq g(k)$. Then the problem decided by $\mathbb{A}$ with oracle $A$ can be decided deterministically by cycling through all $(f(k)p(|x|))^{h(k)}$ sets of guessed integers, simulating the computation of $\mathbb{A}$ for each set of guesses, and simulating, for every query of $\mathbb{A}$ in every computation, the RAM described in the query instance for at most $(f(k)p(|x|))^{g(k)}+g(k)$ steps. A (somewhat tedious) calculation shows that the total number of steps performed by such a simulation can be upper-bounded by $n^{k'}+k'$, for a suitable $n=\poly(|x|)$ and $k'$ that depends computably on $k$ alone, so an $\mathbf{FPT}$-machine with oracle access to $A$ can query the oracle for the deterministic RAM performing the above computation, and thus decide the same problem as $\mathbb{A}$ with oracle $A$.

We proceed with the proof of the second statement. Given any parameterized problem $B$, let $\textsc{1}_B$ be the parameterized problem defined as follows: the instance $(1^n,k)$ is in $\textsc{1}_B$ if and only if there exists an instance $((x_1,\ldots ,x_n),k)$ in $B$, with $x_i\in[n]$ for all $i\in[n]$, such that $x_i\neq 1$ for at most $k$ values $i\in[n]$.

We first describe an $\mathbf{A}[1]$-machine $\mathbb{A}$ that, if given oracle access to $B$, decides $\textsc{1}_B$, for any $B$. $\mathbb{A}$ starts by writing the instance $(1^n,k)$ to its oracle registers, and then enters the nondeterministic phase of its computation. It now guesses $2k$ integers $p_1,\ldots,p_k,i_1,\ldots,i_k\in [n]$, and then modifies the instance encoded in its oracle registers so that, for all $j\in[k]$, the $p_j$-th $1$ in the main part of the instance is changed to $i_j$ ($\leq k$ changes in total). Finally, the machine queries the $B$-oracle, and accepts if and only if the answer is 'yes'.

Next, we define a specific parameterized problem $B$ in such a way that no (deterministic) parameter-restricted RAM with oracle access to $B$ decides $\textsc{1}_B$, because every such machine gives the wrong answer on some input.

Let $(\mathbb{M}_i)_{i\in\mathbb{N}}$ be a sequence of all valid programs of oracle RAMs. Let $l:\mathbb{N}\to\mathbb{N}$ be such that for every $t\in\mathbb{N}$, $|l^{-1}(t)|=\infty$. Let $(K_{i})_{i\in\mathbb{N}}$ be a sequence of functions, with $K_i:\mathbb{N}\to\mathbb{N}$ for all $i$, having the following properties: (i) every $K_i$ has an infinite range; (ii) every positive integer appears in the range of at most one $K_i$; (iii) for every $t\in\mathbb{N}$ such that $K_i^{-1}(t)\neq\emptyset$, $|K_i^{-1}(t)|=\infty$. It is easy to construct an example of a computable function and a computable sequence of functions with the above properties.

Throughout the construction, we will keep track of and use the values $n,j_1,j_2,\ldots$ (initially all are equal to $0$): At stage $i\in\mathbb{N}$, the construction consists of the following steps:
\begin{itemize}
\item[1.] Set $j_{l(i)}:=j_{l(i)}+1$; set $k:=K_{l(i)}(j_{l(i)})$; set $n:=\max\{n,k\}+1$.

\item[2.] Simulate the machine $\mathbb{M}_{l(i)}$ for $n^{k}$ steps on input $(1^n,k)$. On all queries made during the simulation, answer 'no' if an instance has not been queried during any simulated computation up to this point, otherwise answer consistently with previous answers.

\item[3.] If $\mathbb{M}_{l(i)}$ terminates during the simulation and rejects, add to $B$ some instance $((x_1,\ldots,x_{n}),k)$ such that $x_j\in[n]$ for all $j$, $x_j\neq 1$ for at most $k$ values of $j$, and such that the instance has never been queried during any simulation up to this point in the construction. If the machine terminates and accepts, or does not terminate in the required number of steps, do nothing (so that no instances $((x_1,\ldots,x_{n}),k)$ are in $B$, and hence $(1^n,k)\notin\textsc{1}_B$).

\item[4.] Let $n'$ be the largest integer such that a query of the form $((x_{1},\ldots,x_{n'}),k')$, with $x_j\in[n']$ for all $j$,  was made during the simulation at this stage. Set $n:=\max\{n,n'\}$.
\end{itemize}

Let $\mathbb{M}_{l(i)}$ be a deterministic oracle RAM such that, if $\mathbb{M}_{l(i)}$ has access to an oracle for $B$, then, for some computable function $f$ and constant $c\geq 0$, the machine halts on any input $(1^n,k)$ ($n>k$) after at most $f(k)n^c$ steps. By the property of $l$ that it takes every value in its range infinitely many times, it follows that $\mathbb{M}_{l(i)}$ will be run on infinitely many inputs. By property (i) of $(K_j)_{j\in\mathbb{N}}$, $\mathbb{M}_{l(i)}$ will be run on inputs with arbitrarily large parameter values $k$. Thus we may assume that $k>c+1$. By property (iii) of $(K_j)_{j\in\mathbb{N}}$, every value that is the parameter value of an input on which $\mathbb{M}_{l(i)}$ is run, is the parameter value for inputs of $\mathbb{M}_{l(i)}$ at infinitely many stages. Thus we may assume that the parameter value $k>c+1$ is fixed for the machine $\mathbb{M}_{l(i)}$. Finally, because the number $n$ is increased for every simulation, it follows that $\mathbb{M}_{l(i)}$ will be run on inputs with parameter value $k$ and arbitrarily large numbers of ones in the first part of the input. Thus we may assume that $n>f(k)+1$ (since $k$ is now fixed), and hence that $f(k)n^c<(n-1)^k$.

For the input $(1^n,k)$, there are more than $(n-1)^k$ instances $((x_1,\ldots,x_n),k)$ such that $x_j\in[n]$ for all $j$ and $x_j\neq 1$ for exactly $k$ values of $j$, whose membership in $B$ could cause $(1^n,k)$ to be in $\textsc{1}_B$. Since at the beginning of every simulation, $n$ is set so that no instances $((x_1,\ldots,x_n),k)$ have been queried in any previous simulation, it follows that on input $(1^n,k)$, $\mathbb{M}_{l(i)}$ will terminate (because $f(k)n^c<n^k$), and at the end of the simulation, some instances whose membership in $B$ could cause the input to be in $\textsc{1}_B$ will never have been queried. We can therefore ensure that $\mathbb{M}_{l(i)}$ gives the wrong answer on input $(1^n,k)$ (for $n$ and $k$ chosen as above), by either placing an unqueried instance of the right form into $B$ (if the machine's output is `no') or by placing no such instance into $B$ (if the machine's output is `yes'). Since we can make every machine be wrong on some input in this way, we may conclude that $\textsc{1}_B\notin \mathbf{FPT}(B)$.
\end{proof}

\noindent\emph{Relativization results for random access machines with parameter-bounded nondeterminism.}

For this machine model, we first need to define the analogue of the Polynomial Hierarchy.

\begin{defin}
For each $t\geq 1$, let $\mathbf{\Sigma}^{[P]}_t$ be the class of parameterized problems that can be decided by a parameter-restricted $t$-alternating ARAM $\mathbb{A}$ such that, for some computable function $h$, on any input $(x,k)$, $\mathbb{A}$ executes at most $h(k)$ nondeterministic instructions on any computation path. Furthermore, we define $\mathbf{W[P]H}:=\bigcup_{t=1}^\infty \mathbf{\Sigma}^{[P]}_t$.
\end{defin}
Clearly, $\mathbf{W[P]}=\mathbf{\Sigma}^{[P]}_1\subseteq\mathbf{W[P]H}\subseteq\mathbf{AW[P]}$. For $t\geq 2$, $\mathbf{\Sigma}^{[P]}_{t}$-complete problems can be obtained by modifying known $\mathbf{W[P]}$- or $\mathbf{AW[P]}$-complete problems appropriately (see \cite{cfg1} or \cite{fg}).

We turn to the oracle characterization of this hierarchy. Since a $\mathbf{W[P]}$-machine can compute fpt-reductions at any point in the computation, the choice of the complete problem given as an oracle is no longer important.  Now the proof of the theorem proceeds in the same way as the characterization of $\mathbf{PH}$ in terms of oracle machines (see \cite{ab}, Section 5.5), but note that for the ``$\subseteq$''-inclusion, the restrictions on the oracle access are nevertheless essential: balanced access ensures that the $\mathbf{\Sigma}^{[P]}_{t+1}$-machine can nondeterministically decide the instances queried by the oracle machine, and parameter-bounded access ensures that the number of queries made by the oracle machine is not too large for a $\mathbf{\Sigma}^{[P]}_{t+1}$-machine to simulate.
\begin{theor}\label{wp_oracl_thm}
For each $t\geq 1$, we have $\mathbf{W[P]}^{\mathbf{\Sigma}^{[P]}_t}=\mathbf{\Sigma}^{[P]}_{t+1}$.
\end{theor}
\begin{corol}
For any $t,u\geq 1$, if $\mathbf{\Sigma}^{[P]}_t=\mathbf{\Sigma}^{[P]}_{t+u}$, then $\mathbf{W[P]H}=\mathbf{\Sigma}^{[P]}_t$.
\end{corol}

\noindent Finally, we have the oracle separation result for this machine model, as in \cite{bakgs}:
\begin{theor}\label{wp_oraclesep_thm}
There exist parameterized oracles $A$ and $B$ such that
\begin{displaymath}
\mathbf{FPT}(A)_{para}=\mathbf{W[P]}^{A}\textrm{\ \ and\ \ }\mathbf{W[P]}^{B}\setminus\mathbf{FPT}(B)\neq\emptyset.
\end{displaymath}
\end{theor}

\noindent For the proof, it suffices to use the same two oracles as in the proof of Theorem \ref{oraclesep_thm}.

\section{Interactive proof systems for parameterized complexity classes}\label{ip_sec}
A classical interactive proof system consists of a verifier and a prover who exchange messages in order for the verifier to decide whether a given input is a `yes'-instance of a problem. The verifier is a probabilistic TM, meaning that he can guess random bits, but his computation throughout the entire interaction is time-bounded polynomially in terms of the size of the input instance (and therefore so is the length of the messages he can send or receive). The prover is computationally all-powerful, but he only sees the input and the messages sent by the verifier (not the verifier's random bits), and his goal is to convince the verifier to accept. A proof system is said to decide a problem $Q$ if every $x\in Q$ is accepted by the verifier with probability (over the verifier's random bits) $\geq 2/3$ for some prover, and every $x\notin Q$ is accepted by the verifier with probability $\leq 1/3$ for every prover. See \cite{ab}, Chap. 8, for the formal definitions.

Here we make a slight change to this definition, in order to apply the concept in the parameterized setting, by letting the verifier be a probabilistic RAM (meaning that he can guess non-negative integers of bounded size in a single step), and allowing the messages between verifier and prover to be strings of non-negative integers of size bounded in terms of the size of the input and the parameter. This change does not affect the (classical) class $\mathbf{IP}$ (see Remark \ref{input_rem}), but allows us to apply separate bounds to different aspects of the proof systems.

Aside from the number of rounds, proof systems in the class $\mathbf{IP}$ have a number of other parameters that are implicitly bounded polynomially in the size of the instance, due to the requirement that the verifier be a polynomial-time machine. Among them are: the number of random guesses that the verifier can make, the length of the messages that the prover can send to the verifier, and the length of the verifier's computations between two messages. When considering interactive proof systems for parameterized problems, any one of these numbers can be bounded either computably in the parameter value of the instance, or have a bound of the form $f(k)p(|x|)$, where $f$ is a computable function and $p$ a polynomial function.\\

\noindent\emph{Arithmetization of first-order formulas with relational vocabularies.}

Before we can give interactive proof systems for parameterized complexity classes, we need to adapt the main technical tool used in such results, namely arithmetization.

Let $\mathcal{A}$ be a relational structure with universe $A=\{0,\ldots,u\}$ ($u\geq 1$), and let $\phi=$\\$\exists x_1\forall x_2\ldots Qx_k\psi(x_1,\ldots,x_k)$ be a first-order formula of the same vocabulary as $\mathcal{A}$, where $\psi$ is quantifier-free. Let $q$ be the smallest prime between $u+1$ and $2(u+1)$. We identify $A$ with a subset of $GF(q)$ in the obvious way. We show how to define a multivariate polynomial $P_{\mathcal{A},\psi}\in GF(q)[X_1,\ldots,X_k]$ such that $\forall(a_1,\ldots,a_k)\in A^k$:  $P_{\mathcal{A},\psi}(a_1,\ldots,a_k)=\psi(a_1,\ldots,a_k)$, and how to define operations $\exists X_i$ and $\forall X_i$ such that $\exists X_1\forall X_2\ldots QX_k P_{\mathcal{A},\psi}=\phi$.

We start with the atomic formulas and define $Eq(X,Y):=1-(X-Y)^{q-1}$. By Fermat's Little Theorem, $(x-y)^{q-1}=1$ (in $GF(q)$) whenever $x-y\in GF(q)\setminus\{0\}$. Thus, $Eq(x,y)=1$ if $x=y$, and $Eq(x,y)=0$ if $x\neq y$, for all $x,y\in GF(q)$. The relations in $\mathcal{A}$ can similarly be transformed into polynomials, and we illustrate this for a ternary relation $T$: Define
\begin{displaymath}
P_T(X,Y,Z):=\sum_{(u,v,w)\in T}(1-(X-u)^{q-1})(1-(Y-v)^{q-1})(1-(Z-w)^{q-1}).
\end{displaymath}
By the same reasoning as above, we have that $P_T(u,v,w)=T(u,v,w)$ for all $u,v,w\in A$. The degree of both $Eq$ and $P_T$ in each of their variables is $q-1$, but it could happen that $T$ appears in an FO formula, for example, as ``$Txxz$'', in which case using the above representation would result in a polynomial expression whose degree in $X$ is greater than $q-1$. But such an increase of the degree is unnecessary, because the following polynomial also represents the atomic formula $Txxz$: $P_T(X,Z):=\sum_{(u,u,w)\in T}(1-(X-u)^{q-1})(1-(Z-w)^{q-1})$. We may therefore assume that polynomial representations of relations of any arity have degree $q-1$ in each of their variables.

Let $\alpha=\alpha(x_1,\ldots,x_l)$ and $\beta$ be (not necessarily quantifier-free) formulas with the same vocabulary as $\mathcal{A}$, for which polynomials $P_{\mathcal{A},\alpha}$ and $P_{\mathcal{A},\beta}$ have been defined. Then we associate the polynomial $P_{\mathcal{A},\alpha}\cdot P_{\mathcal{A},\beta}$ with $\alpha\wedge\beta$, $P_{\mathcal{A},\alpha}+ P_{\mathcal{A},\beta} - P_{\mathcal{A},\alpha}\cdot P_{\mathcal{A},\beta}$ with $\alpha\vee\beta$, and $1-P_{\mathcal{A},\alpha}$ with $\neg \alpha$. Furthermore, we associate with the formulas $\forall x_i\alpha(x_1,\ldots,x_l)$ and $\exists x_i\alpha(x_1,\ldots,x_l)$, respectively, the polynomials
\begin{align*}
\forall X_iP_{\mathcal{A},\alpha}(X_1,\ldots,X_l):=&\prod_{z=0}^{u}P_{\mathcal{A},\alpha}(X_1,\ldots,X_{i-1},z,X_{i+1},\ldots,X_l)\textrm{\ and}\\
\exists X_iP_{\mathcal{A},\alpha}(X_1,\ldots,X_l):=&1-\prod_{z=0}^{u}(1-P_{\mathcal{A},\alpha}(X_1,\ldots,X_{i-1},z,X_{i+1},\ldots,X_l)).
\end{align*}
Since the definitions of the $\forall$ and $\exists$ operations for polynomials involve products, the degree of the resulting polynomial in each variable can be larger by a factor of $u$ than that of the polynomial that the operation is applied to. However, when evaluating such a polynomial on values in $GF(q)$, we have, again by Fermat's Little Theorem, that $x^r=x^{r-q+1}$ for all $r\geq q$. It is therefore possible to reduce the degree of a polynomial in each variable to at most $q-1$, without changing its value on any tuple of elements from $GF(q)$. We define a \emph{degree reduction operation} that produces a polynomial whose degree in one variable is exactly $q-1$, and which coincides with the original polynomial on all tuples of values from $A$.
\begin{displaymath}
\R X_iP_{\mathcal{A},\alpha}(X_1,\ldots,X_l):=\sum_{z=0}^{u}(1-(X_i-z)^{q-1})P_{\mathcal{A},\alpha}(X_1,\ldots,X_{i-1},z,X_{i+1},\ldots,X_l).
\end{displaymath}

We will want to evaluate the polynomials associated with first-order formulas on values from some larger field, while preserving their properties on $A$ (which we identified with a subset of $GF(q)$). If we were to simply identify the numbers $0,\ldots,u$ with elements from $GF(q')$, for some prime $q'>q$, while leaving the definitions of the polynomials unchanged (with $q-1$ in the exponents), then Fermat's Little Theorem would no longer apply, and it would no longer be the case that the polynomials take values in $\{0,1\}$ when evaluated on tuples of values from $A$. In order to preserve this property without increasing the degree of the polynomials we work with, we choose the field $GF(q^h)$, for some reasonable value $h$, which contains $GF(q)$ as a subfield. The elements of $GF(q^h)$ can themselves be identified with $h$-tuples of elements of $GF(q)$, and, as such, arithmetic operations in $GF(q^h)$ can be computed efficiently, as long as $h$ is not too large. Note that for any $P\in GF(q^h)[X_1,\ldots,X_l]$, $\R X_iP$ coincides with $P$ on all tuples of values from $A$.
\\\\
With arithmetization generalized in this way, we are now in a position to construct an IP similar to the one used in \cite{shen} to show that $\mathbf{PSPACE}\subseteq\mathbf{IP}$, and prove the following:

\begin{theor}\label{awsat_thm}
For every problem $Q\in\mathbf{AW[SAT]}$, there is an interactive proof system deciding $Q$ such that, for some computable functions $f$ and $h$, and a polynomial $p$, on any input $(x,k)$, the verifier runs in time $f(k)p(|x|)$ and makes at most $h(k)$ random guesses, and the interaction has at most $h(k)$ rounds.
\end{theor}
\begin{proof}
Let $Q$ be a problem in $\mathbf{AW[SAT]}$. Then $Q\leq^{\textrm{fpt}}p\textrm{-}var\textsc{-MC}(\textrm{PNF})$ and hence there exists an algorithm $\mathbb{R}$, such that for some computable functions $f$ and $h$, and a polynomial function $p$, we have that:

-for any instance $(x,k)$, $\mathbb{R}(x,k)=:((\mathcal{A},\phi),k')\in p\textrm{-}var\textsc{-MC}(\textrm{PNF}) \Leftrightarrow (x,k)\in Q$;

-the algorithm $\mathbb{R}$ runs in time $f(k)p(|x|)$;

-$k'\leq h(k)$.\\
Note that, due to the restriction on the running time of $\mathbb{R}$, we have that $||\mathcal{A}||+|\phi|+k'\leq f(k)p(|x|)$.

Let $f'(k):=\max\{f(k),(h(k)^2+3h(k))/2\}$. Let $q$ be the smallest prime between $f'(k)p(|x|)+1$ and $ 2(f'(k)p(|x|)+1)$. We describe an interactive proof system with $O(h(k)^2)$ rounds, such that the verifier runs in time $O(poly(q))$, accepts with probability close to 1 if $(x,k)\in Q$, and rejects with probability at least $2/3$ if $(x,k)\notin Q$.

We may assume that $\phi=\exists x_1\forall x_2\ldots Qx_{h(k)}\psi(x_1,\ldots,x_{h(k)})$, where $\psi$ is quantifier-free and of size $\leq q$. Therefore, we may also assume that $P_{\mathcal{A},\psi}$ can be written as an algebraic expression of size at most $\leq O(q^2)$ (since every relation symbol in $\psi$ must be replaced with a polynomial which can be written using at most $O(||\mathcal{A}||)\leq O(q)$ symbols). It follows then that $P_{\mathcal{A},\psi}$ can be evaluated on tuples of values from $GF(q^4)$ in time polynomial in $q$.

Evidently $\phi$ is satisfiable if and only if $P_{\mathcal{A},\phi}=1$. Now that the problem of deciding the satisfiability of $\phi$ has been reduced to testing whether a polynomial identity holds, we can in principle use the same proof system as the one for the $\mathbf{PSPACE}$-complete problem $\textsc{TQBF}$ in \cite{shen}, except that the linearization operation from that proof is replaced by our degree reduction operation $\R$. For the sake of completeness we describe the proof system here.

Since the $\exists$ and $\forall$ operations set the variable they are applied for only to values in $GF(q)$, an arbitrary number of $\R$ operations can be added to the expression defining $P_{\mathcal{A},\phi}$ without changing the result. In other words, $P_{\mathcal{A},\phi}=1$ if and only if
\begin{equation}\label{Req}
\exists X_1 \R X_1 \forall X_2 \R X_1 \R X_2 \exists X_3 \ldots QX_{h(k)}\R X_1 \R X_2 \ldots \R X_{h(k)}P_{\mathcal{A},\psi}(X_1,\ldots,X_{h(k)})=1.
\end{equation}
We define a number of polynomials based on the expression on the left-hand side of (\ref{Req}), by successively removing operations on variables from left to right. Thus, each of the polynomials has non-zero degree in some variable $X_i$ if and only if in the expression defining the polynomial the single $\exists$/$\forall$ operation on that variable in (\ref{Req}) has been removed.
\begin{align*}
&P_0:=\exists X_1 \R X_1 \forall X_2 \R X_1 \R X_2 \exists X_3 \ldots QX_{h(k)}\R X_1 \R X_2 \ldots \R X_{h(k)}P_{\mathcal{A},\psi}(X_1,\ldots,X_{h(k)}).\\
&P_1(X_1):=\R X_1 \forall X_2 \R X_1 \R X_2 \exists X_3 \ldots QX_{h(k)}\R X_1 \R X_2 \ldots \R X_{h(k)}P_{\mathcal{A},\psi}(X_1,\ldots,X_{h(k)}).\\
&P_2(X_1):=\forall X_2 \R X_1 \R X_2 \exists X_3 \ldots QX_{h(k)}\R X_1 \R X_2 \ldots \R X_{h(k)}P_{\mathcal{A},\psi}(X_1,\ldots,X_{h(k)}).\\
&P_3(X_1,X_2):=\R X_1 \R X_2 \exists X_3 \ldots QX_{h(k)}\R X_1 \R X_2 \ldots \R X_{h(k)}P_{\mathcal{A},\psi}(X_1,\ldots,X_{h(k)}).\\
&\vdots\\
&P_{(h(k)^2+3h(k))/2-1}(X_1,\ldots,X_{h(k)}):=\R X_{h(k)}P_{\mathcal{A},\psi}(X_1,\ldots,X_{h(k)}).\\
&P_{(h(k)^2+3h(k))/2}(X_1,\ldots,X_{h(k)}):=P_{\mathcal{A},\psi}(X_1,\ldots,X_{h(k)}).
\end{align*}
For every $t\in\{1,\ldots,(h(k)^2+3h(k))/2\}$, $P_{t-1}(X_1,\ldots,X_i)=QX_jP_t(X_1,\ldots,X_j)$, where $Q\in\{\exists,\forall,\R\}$. If $Q$ is $\exists$ or $\forall$, then in the last identity $i=j-1$, otherwise $i=j$. Due to the interspersed $\R$ operations and the fact that the degree of $P_{\mathcal{A},\psi}(X_1,\ldots,X_{h(k)})$ in each variable is $\leq q^2$, we have that $P_t(X_1,\ldots,X_j)$ has degree at most $q^2$ in $X_j$.
\\\\
\emph{The proof system.}

Before the start of the interaction, the verifier runs a probabilistic algorithm for finding an irreducible polynomial of degree $4$ with coefficients in $GF(q)$, which enables him to perform computations in $GF(q^4)$ \cite{rab80}. He sends the coefficients of this polynomial to the prover.

The prover tries to convince the verifier that $P_0=1$. The verifier can not efficiently evaluate $P_0$ by himself, but can efficiently evaluate the last of the above polynomials on any tuple of values from $GF(q^4)$.

In round $t$ of the interaction, for $t\in\{1,\ldots,(h(k)^2+3h(k))/2\}$, the prover attempts to convince the verifier that $P_{t-1}(a_1,\ldots,a_i)=s_{t-1}$, where $a_1,\ldots,a_i,s_{t-1}\in GF(q^4)$ are values chosen at previous rounds (except $s_0=1$). To do this, he must send the coefficients of a polynomial $S(X_j)$ of degree $\leq q^2$, which is claimed to be $P_t(a_1,\ldots,a_{j-1},X_j)$. The verifier computes $QX_j S(X_j)$ in time polynomial in $q$. We have two cases:

1. If $Q$ is $\exists$ or $\forall$, and hence $i=j-1$, then $QX_j S(X_j)$ is a constant that must be equal to $s_{t-1}$. The verifier checks this and rejects if the equality does not hold.

2. If $Q=\R$, and hence $i=j$, then $QX_j S(X_j)$ is a polynomial in $X_j$. The verifier checks whether $(QX_j S(X_j))(a_i)=s_{t-1}$ and rejects if this is not the case.

In both cases the verifier then chooses uniformly at random an element $a\in GF(q^4)$, sends $a_j:=a$ to the prover and sets $s_t:=S(a)$. They proceed to the next round (the verifier delays his final decision until some later round, since he has found no reason to reject in this round).

If the interaction reaches round $t=(h(k)^2+3h(k))/2+1$, then the verifier has values $a_1,\ldots,a_{h(k)},s_{t-1}\in GF(q^4)$ chosen during previous rounds. He checks whether $P_{\mathcal{A},\psi}(a_1,\ldots,a_{h(k)})$\\$=s_{t-1}$ and accepts if this is the case, otherwise he rejects.
\\\\
\emph{Analysis of the proof system.}

The verifier can run the procedure for obtaining an irreducible polynomial some constant number of times, so that the probability of success is $>1-1/100$. If he fails to find a suitable polynomial, he rejects.

If $P_0=1$, and the verifier has not rejected in the beginning, then the prover can always cause the verifier to accept, by sending the correct polynomial each round. Assume then that $P_0=0$. Since the claim that $P_{t-1}(a_1,\ldots,a_i)=s_{t-1}$ is false in round $1$, while the only way for the verifier to accept is if this claim is true in the last round, there must be some $t\geq 1$ such that $P_{t-1}(a_1,\ldots,a_i)\neq s_{t-1}$, but $P_{t}(a_1,\ldots,a_{j-1},a)= s_{t}$. Thus, in round $t$, the prover must produce a polynomial $S(X_j)\neq P_{t}(a_1,\ldots,a_{j-1},X_j)$ such that $S(a)=P_{t}(a_1,\ldots,a_{j-1},a)$ for the randomly chosen $a\in GF(q^4)$. Note that $S(X_j)- P_{t}(a_1,\ldots,a_{j-1},X_j)$ has degree $\leq q^2$, and therefore has at most $q^2$ roots in $GF(q^4)$. Thus, the probability that $s_t=S(a)=P_{t}(a_1,\ldots,a_{j-1},a)$ holds for a randomly chosen $a\in GF(q^4)$, is at most $q^2/q^4$. By the union bound, the probability of this happening in at least one round is no greater than $\frac{((h(k)^2+3h(k))/2)\cdot q^2}{q^4}\leq \frac{q^3}{q^4}$. We may assume that $q\geq 4$, in which case this probability is $< 1/3$.
\end{proof}

The IP in Theorem \ref{awsat_thm} has both the number of rounds and the number of random guesses made by the verifier bounded computably in terms of the parameter, but the length of the prover's messages and of the verifier's computations between rounds are ``fpt-bounded''. In order for an $\mathbf{AW[*]}$-machine to simulate an interactive proof, it would presumably need to nondeterministically guess the prover's messages, as well as the random guesses made by the verifier, so the entire interaction would have to be simulated in the last $h(k)$ steps of the computation (due to tail-nondeterminism). In other words, the proof system would have to be such that the verifier only performs an fpt-bounded pre-computation, followed by an interaction that is entirely bounded in the parameter alone. We conjecture that the class of problems with such IPs, which we call $\mathbf{IP}^{tail}$, is precisely $\mathbf{AW[*]}$. The evidence for this conjecture is that when the size of the FO formula is bounded in terms of the parameter, it seems that the IP from Theorem \ref{awsat_thm} can be improved so that at least the length of the prover's messages depends only on the parameter, by using only symbols for the polynomials representing the atomic relations, rather than expanding them into algebraic expressions. Getting the same bound for the verifier's computations between rounds is more challenging.

\section{Conclusions}
We have shown that, with some degree of effort, certain classical methods can be put to use in the parameterized setting, although some theorems only partially transfer over. The fact that different aspects of the computation of a RAM are bounded differently, and that some computational resources can be tail-restricted, ensures that the machine-based theory of parameterized intractability is by no means just ``complexity theory with RAMs''.

One can now attempt to make some progress on the problem of separating matching levels of the $\mathbf{A}$- and the $\mathbf{W}$-Hierarchy, by proving oracle separations when reasonable restrictions are placed on the oracle access of the respective machines.

\subsection*{Acknowledgments}
The author is grateful to Yijia Chen and S\'andor Kisfaludi-Bak for helpful discussions, to Martin Bottesch, S\'andor Kisfaludi-Bak, and Ronald de Wolf for comments on a draft of this paper, and to two anonymous referees for helpful comments on the version submitted to IPEC 2017.


\begin{thebibliography}{00}
\bibitem{abrah} K.A. Abrahamson, R.G. Downey, and M.R. Fellows. Fixed-parameter
tractability and completeness IV: On completeness for W[P] and PSPACE
analogs. In \emph{Annals of Pure and Applied Logic}, Vol. 73, pp. 235–276, 1995.


\bibitem{ab} S. Arora, B. Barak. Computational Complexity: A Modern Approach. Cambridge, 2009.

\bibitem{aw08} S. Aaronson, A. Wigderson. Algebrization: A new barrier in complexity theory. In \emph{ACM Trans. Comput. Theory}, Vol. 1(1), 2009.

\bibitem{bakgs} T. Baker, J. Gill, and R. Solovay. Relativizations of the P=?NP question. In \emph{SIAM J. Comput.}, Vol 4(4), pp. 431-442, 1975.

\bibitem{cfg2} Y. Chen, J. Flum. Machine characterizations of the classes of the W-hierarchy. In \emph{Proceedings of the 17th International Workshop on Computer Science Logic, Lecture Notes in Computer Science}, Vol. 2803, Springer, Berlin, pp. 114-127, 2003.

\bibitem{cfg1} Y. Chen, J. Flum, M. Grohe. Bounded nondeterminism and alternation in parameterized complexity theory. In \emph{Proceedings of the 18th IEEE Conference on Computational Complexity}, pp. 13-29, 2003.

\bibitem{cfg3} Y. Chen, J. Flum, M. Grohe. Machine-based methods in parameterized complexity theory. In \emph{Theor. Comput. Sci.}, Vol. 339, pp. 167-199, 2005.

\bibitem{df92} R.G. Downey and M.R. Fellows. Fixed-parameter tractability and completeness III - Some structural aspects of the W hierarchy. In \emph{Complexity Theory} (ed. K. Ambos-Spies, S. Homer, and U. Sch\"oning), Cambridge University Press, pp. 166-191, 1993.

\bibitem{df1} R.G. Downey, M.R. Fellows. Parameterized Complexity. Springer, Berlin, 1999.

\bibitem{df2} R.G. Downey, M.R. Fellows. Fundamentals of Parameterized Complexity. Springer, 2013.

\bibitem{fg2001} J. Flum, M. Grohe. Fixed-parameter tractability, definability, and model checking. In \emph{SIAM J. Comput.}, Vol 31(1), pp. 113-145, 2001.

\bibitem{fg} J. Flum, M. Grohe. Parameterized Complexity Theory. Springer, Berlin, 2006.

\bibitem{fortnow} L. Fortnow. A simple proof of Toda's theorem. In \emph{Theory of Computing}, Vol. 5, pp. 135-140, 2009.

\bibitem{muell} J.A. Montoya, M. M\"uller. Parameterized random complexity. In \emph{Theory. Comput. Syst.}, Vol. 52, pp. 221-270, 2013.

\bibitem{papadim} C.H. Papadimitriou. Computational Complexity. Addison-Wesley, 1994.

\bibitem{rab80} M.O. Rabin. Probabilistic algorithms in finite fields. In \emph{SIAM J. Comput.}, Vol. 9, pp. 273-280, 1980.

\bibitem{sham} A. Shamir. IP = PSPACE. In \emph{J. ACM}, Vol. 39(4), pp. 869-877, 1992.

\bibitem{shen} A. Shen. IP = PSPACE: simplified proof. In \emph{J. ACM}, Vol. 39(4), pp. 878-880, 1992.

\bibitem{toda} S. Toda. PP is as hard as the polynomial-time hierarchy. \emph{SIAM J. Comput.}, Vol. 20(5), pp. 865-877, 1991.
\end{thebibliography}
\end{document}